\documentclass[onefignum,onetabnum,hidelinks]{siamonline220329}

\usepackage{lipsum}
\usepackage{amsfonts}
\usepackage{graphicx}
\usepackage{epstopdf}
\usepackage{fancyvrb}
\usepackage{listings}
\usepackage{algorithmic}
\usepackage[multiple]{footmisc}
\ifpdf
  \DeclareGraphicsExtensions{.eps,.pdf,.png,.jpg}
\else
  \DeclareGraphicsExtensions{.eps}
\fi

\usepackage{enumitem}
\setlist[enumerate]{leftmargin=.5in}
\setlist[itemize]{leftmargin=.5in}

\newcommand\nonumberfootnote[1]{%
  \begingroup
  \renewcommand\thefootnote{}\footnote{#1}%
  \addtocounter{footnote}{-1}%
  \endgroup
}

\newsiamremark{remark}{Remark}
\newsiamremark{hypothesis}{Hypothesis}
\crefname{hypothesis}{Hypothesis}{Hypotheses}
\newsiamthm{claim}{Claim}

\headers{Polynomial Preconditioners for Inverse Problems}
{S. S. Iyer, F. Ong, X. Cao, C. Liao, L. Daniel, J. I. Tamir, and K. Setsompop}

\title{
  Polynomial Preconditioners for Regularized Linear Inverse Problems
  \thanks{Submitted to the editors \today.
  \funding{This work was supported in part by the National Institute of Health
           under Grants R01EB020613, R01MH116173, R01EB019437, U01EB025162, and
           P41EB030006, and in part by GE Healthcare.}}
}

\author{
    Siddharth S. Iyer%
    \thanks{Department of Electrical Engineering and Computer Science,
    Massachusetts Institute of Technology, Cambridge, MA
    (\email{ssi.compsens@gmail.com, dluca@mit.edu}).}
  \and
    Frank Ong%
    \thanks{Department of Electrical Engineering,
    Stanford University, Stanford, CA
    (\email{frankongh@gmail.com, kawins@stanford.edu}).}
  \and
    Xiaozhi Cao%
    \thanks{Department of Radiology,
    Stanford University, Stanford, CA
    (\email{xiaozhic@stanford.edu, cyliao@stanford.edu}).}
  \and
    Congyu Liao%
    \footnotemark[4]
  \and
    Luca Daniel%
    \footnotemark[2]
  \and
    Jonathan I. Tamir$^\star$%
    \thanks{Departments of Electrical and Computer Engineering, Diagnostic
    Medicine, and Oden Institute, University of Texas at Austin, Austin,
    TX
    (\email{jtamir@utexas.edu}).}
  \and
    Kawin Setsompop$^\star$%
    \footnotemark[3]
    \footnotemark[4]
}

\usepackage{amsopn}

\begin{document}

\maketitle

\nonumberfootnote{\hspace{-0.275cm}$^\star$J. I. Tamir and K. Setsompop
contributed equally to this article.}
\nonumberfootnote{\hspace{-0.125cm}Parts of this work have been presented at
the ISMRM Annual Conference of 2022 \cite{iyer2022}.}

\begin{abstract}
This work aims to accelerate the convergence of proximal gradient methods used
to solve regularized linear inverse problems.
This is achieved by designing a polynomial-based preconditioner that targets
the eigenvalue spectrum of the normal operator derived from the linear
operator.
The preconditioner does not assume any explicit structure on the linear
function and thus can be deployed in diverse applications of interest.
The efficacy of the preconditioner is validated on three different Magnetic
Resonance Imaging applications, where it is seen to achieve faster iterative
convergence while achieving similar reconstruction quality.
\end{abstract}



\section{Introduction}

Linear inverse problems, particularly those derived from compressed sensing
formulations, are typically posed as the following optimization problem
\cite{foucart2013, chen2001}:
\begin{equation}
x^\star = \left\{
\begin{aligned}
& \underset{x}{\text{argmin}}
& & g(x) \\
& \text{subject to}
& & Ax = b
\end{aligned}\right.
\label{eq:cs}
\end{equation}
Here, $A: \mathbb{C}^n \rightarrow \mathbb{C}^m$ is the linear forward model or
the measurement matrix, $b \in \mathbb{C}^m$ is the acquired data or
measurement vector and $g: \mathbb{C}^n \rightarrow \mathbb{R}$ is a prior
regularization function, often the $l_1-$norm, to enforce sparsity.
In practice, due to measurement errors (such as noise), the following
optimization problem is solved instead
\cite{foucart2013, lustig2007, block2007, liu2009, fessler2020}:
\begin{equation}
x^\star = \left\{
\begin{aligned}
& \underset{x}{\text{argmin}}
& & g(x) \\
& \text{subject to}
& & \lVert Ax - b \rVert_2 \leq \epsilon
\end{aligned}\right.
\label{eq:copt}
\end{equation}
The optimization \cref{eq:copt} is often re-cast as the following
unconstrained problem for ease-of-computation:
\begin{equation}
x^\star = \left\{
\begin{aligned}
& \underset{x}{\text{argmin}}
& & \frac{1}{2} \lVert Ax - b \rVert_2^2 + \lambda g(x)
\end{aligned}\right.
\label{eq:ucopt}
\end{equation}
The Karush-Kuhn-Tucker (KKT) conditions can be used to verify that for any
$\epsilon$, there is an appropriate choice of $\lambda$ such that the solutions
of \cref{eq:copt} and \cref{eq:ucopt} coincide.

The unconstrained formulation \cref{eq:ucopt} is commonly utilized for solving
ill-conditioned least squares problems, where the regularization $(g)$ helps
stabilize the solution.
When $g$ is the $l_1-$norm, the unconstrained formulation \cref{eq:ucopt} is
commonly called referred to as ``LASSO''\cite{tibshirani1996}.

\subsection{Motivation}
In computational Magnetic Resonance Imaging (MRI), iterative proximal methods
\cite{parikh2014} have emerged as the workhorse algorithms to solve linear
inverse problems that are posed in the form of \cref{eq:ucopt}
\cite{lustig2007, lustig2008, block2007, liu2009, fessler2020}.
This is arguably because:
\begin{enumerate}
\item These algorithms are highly generalizable as they do not impose any
      restrictions on the utilized regularization function $(g)$ (as long as
      the proximal operator of $(g)$, defined in the sequel, can be
      calculated).
\item These algorithms leverage matrix-free implementations of the
      forward-model $A$ (i.e., the individual coordinates or entries of $A$ are
      not known or stored in memory) for computationally efficient processing
      of high-dimensional problems.
\end{enumerate}

Due to the above mentioned advantages, MRI reconstructions are typically posed
as \cref{eq:ucopt} and are solved using the Fast Iterative
Shrinkage-Thresholding Algorithm (FISTA)\cite{beck2009}, which enjoys
theoretically optimal convergence.
However, the iterative convergence of FISTA and other iterative proximal
methods is largely limited by the conditioning of $A$ as determined by the
eigenvalues of $A^*A$, where $A^*$ is the adjoint of $A$.
This limits the integration of MRI applications with an ill-posed $A$ that
leverage \cref{eq:ucopt} into clinical practice, as the ill-conditioning
results in long reconstruction times.

\subsection{Contributions}

This work proposes a generalizable polynomial-based preconditioner for faster
iterative convergence of regularized linear inverse problems that leverage
proximal gradient methods (like FISTA).
The evolution of the iterates in proximal gradient descent (PGD)
\cite{parikh2014} is analyzed and a cost function for polynomial optimization
is derived such that the optimized polynomial $(p)$, when utilized as a
spectral function, directly improves the convergence rate of PGD (and
subsequently, FISTA).
The proposed preconditioner $P$ is evaluated as $P = p(A^*A)$.
Once the polynomial is calculated, it can be applied to any application with
forward model $A$ as long as the maximum eigenvalue of $A^*A$ can be estimated.
(Note that standard FISTA also requires an estimate of the maximum eigenvalue
of $A^*A$.)
The proposed preconditioner does not assume any structure on $A$, and
can leverage matrix-free implementations of $A^*A$.
Similarly, the proposed preconditioner does not assume any additional structure
on $g$ other than that its proximal operator, defined in the sequel, can be
evaluated.
Thus, the proposed preconditioner is highly generalizable and can be applied
to various regularized linear inverse problems of interest.
In particular, the proposed method retains the ``plug-and-play'' property of
FISTA, and is thus an ideal candidate for integration into reconstruction
toolboxes like SigPy \cite{ong2019}, MIRT \cite{mirt} and BART
\cite{uecker2015} whose users typically leverage proximal algorithms in an
``out-of-the-box'' manner.

\subsection{Related Works}

A polynomial-based preconditioner for accelerating FISTA, which will be
denoted ``IFISTA'' in this work, was presented in \cite{bhotto2015}.
IFISTA uses a polynomial with binomial coefficients to construct the
preconditioner by evaluating the polynomial on $A^*A$.
In contrast, the polynomial coefficients used in this work are obtained from
optimizing a cost function derived from analyzing the error propagation over
iterations of PGD which, for a fixed polynomial degree, results in polynomial
coefficients that more explicitly target the improved convergence compared to
the binomial coefficients used in IFISTA.
In fact, the cost function utilized in this work has a natural connection to
\cite{johnson1983}, where a similar cost function was proposed as a means of
accelerating the convergence of the Conjugate Gradient.
While \cite{johnson1983} focuses explicitly on ordinary least squares
without regularization and motivates the polynomial design as a means of
approximating the inverse of $A^*A$, this work arrives at the same cost
function for polynomial design, but from the completely different perspective
of using PGD to solve ill-posed problems that leverage regularization.

Jacobi-like ``left-preconditioning'' methods have also been proposed as a
means of accelerating convergence, where a matrix $D$ is designed so that
$A^*DA$ is better conditioned.
$D$ is then incorporated into the least-squares cost in \cref{eq:ucopt}
to yield the following \cite{ong2020}:
\begin{equation}
x^\star = \left\{
\begin{aligned}
& \underset{x}{\text{argmin}}
& & \frac{1}{2} \lVert D^{1/2}(Ax - b) \rVert_2^2 + \lambda g(x)
\end{aligned}\right.
\label{eq:ducopt}
\end{equation}
For non-Cartesian MRI, these methods are typically called ``Density
Compensation'' and have seen wide adoption \cite{pipe1999, pruessmann2001}.
Note that by defining $A_{\,\text{DCF}}$ and $b_{\,\text{DCF}}$ as $D^{1/2}A$
and $D^{1/2}b$ respectively, \cref{eq:ducopt} is a special case of
\cref{eq:ucopt}.
Thus, general methods like FISTA, IFISTA and the proposed preconditioner can
synergistically leverage such structure-based left-preconditioning.

In the intersection of computational MRI and regularized linear inverse
problems, several preconditioning methods have been proposed that leverage
the circulant structure of $A$, such as in \cite{koolstra2019, ramani2011,
weller2014, muckley2016}.
In \cite{ong2020}, a Frobenius-norm-optimized diagonal preconditioner for the
dual variables of the primal-dual hybrid gradient (PDHG) algorithm
\cite{chambolle2011} was presented to improve convergence for non-Cartesian MRI
applications.
However, by utilizing explicit structure, it is unclear on whether the
mentioned methods will generalize well to an arbitrary $A$.
This increases the barriers to entry of such preconditioning methods,
particularly for novel applications.
This is particularly true for recent contrast-resolved MRI
applications\cite{cao2022, tamir2017, wang2019, ma2013, liang2007, huang2013,
zhao2015, velikina2015, ben2015, kecskemeti2016, velikina2013, cao2018} with
large measurement sizes and an ill-conditioned forward operator $A$ such as
Echo Planar Time Resolved Imaging \cite{wang2019} and Magnetic Resonance
Fingerprinting (MRF) \cite{ma2013, cao2018, cao2022}.
The high dimensionality and ill-conditioning of $A$ may also result in
computationally intensive procedures to estimate the preconditioner, such as
with the Frobenius norm formulation used in \cite{ong2020} and circulant
preconditioner design process proposed in \cite{muckley2016}.

Another advantage of FISTA, IFISTA and the proposed preconditioner is that
they can utilize efficient implementations of the normal operator $(A^*A)$,
which cannot be leveraged by PDHG and \cite{ong2020}.
For example, non-Cartesian MRI reconstructions can leverage the Toeplitz
structure of the normal operator of the non-uniform Fourier transform
to avoid expensive gridding operations \cite{wajer2001, fessler2005,
baron2018}.
Temporal subspace methods such as $T_2-$Shuffling can use the
``spatio-temporal'' kernel to avoid expanding into the ``echo'' dimension at
each iteration to significantly reduce the number of Fast Fourier Transforms
needed\cite{tamir2017}.

In applications where the regularization $(g)$ is cognate to the $l_1-$norm to
enforce sparsity, algorithms that accelerate Iteratively Re-Weighted Least
Squares (IRWLS) \cite{daubechies2010, foucart2013} have been proposed such as
\cite{xu2018} and \cite{chen2018}.
However, since these algorithms do not use proximal operators, it limits the
different types of regularizations that can be tested without significantly
modifying the chosen algorithm.
For example, should a user choose to use the algorithm proposed in
\cite{chen2018}, it is on the user to verify whether their operator $A$ is
``diagonally dominant'' and that their chosen regularization $(g)$ yields
an easy-to-calculate preconditioner (which is re-calculated at every outer
iteration).

In the spirit of generalizability and to retain the simplicity and the
``plug-an-play'' benefits of FISTA, this work compares and contrasts the
iterative convergence of the proposed preconditioner against other
generalizable methods like FISTA, IFISTA and the Alternating Direction Method
of Multipliers (ADMM) \cite{parikh2014}.

\section{Theory\label{sec:theory}}
Without loss of generality, let the induced norm of $A$ be unitary, which then
implies the eigenvalues of $A^*A$ and $AA^*$ lie in the interval $[0, 1]$.
Let $\langle \cdot, \cdot \rangle$ denote the standard Euclidean complex inner
product, and let $\lVert \cdot \rVert_p$ denote the standard Euclidean
$l_p-$norm.

The sequel will focus on the unconstrained formulation \cref{eq:ucopt}
as the experiments leverage the smoothness of the least squares
term in the objective function for faster convergence via Nesterov
acceleration \cite{beck2009, nesterov2018}.
\begin{remark}
For simplicity, the following analysis focuses on traditional
PGD.
That being said, the improved conditioning of the method naturally
translates into faster convergence when using FISTA as well.
\end{remark}

\subsection{Background\label{sec:background}}
Following \cite{parikh2014}, the proximal operator of a closed and proper
convex function $g$ is defined as:
\begin{definition}[Proximal Operator]
\begin{equation}
\text{\normalfont \bf prox}_{\alpha g} (v) = \left\{
\begin{aligned}
& \underset{x}{\text{\normalfont argmin}}
& & \frac{1}{2} \lVert x - v \rVert_2^2 + \alpha g(x)
\end{aligned}\right.
\label{eq:proxdef}
\end{equation}
\end{definition}

This proximal operator is used in the Proximal Gradient Descent (PGD)
algorithm to solve optimization algorithms of the form \cref{eq:ucopt}
\cite{parikh2014}.
With the assumption that the induced norm of $A$ is unitary, the PGD
algorithm to solve \cref{eq:ucopt} is listed as \cref{alg:pgd}.
\begin{algorithm}
\caption{Proximal Gradient Descent \cite{parikh2014}}
\label{alg:pgd}
\begin{algorithmic}
\STATE{\textit{Inputs:}\begin{itemize}
  \item[-] forward model $A$
  \item[-] measurements $b$
  \item[-] proximal operator $\text{\bf prox}_{\lambda g}$
  \item[-] $x_0 = 0$
  \end{itemize}}
\STATE{\textit{Step k:} $(k \geq 0)$ Compute
\begin{equation}
x_{k+1} = \text{\bf prox}_{\lambda g} \left(x_k - A^*(Ax_k - b)\right)
\label{eq:pgdupdate}\end{equation}}
\end{algorithmic}
\end{algorithm}

The iterations of \cref{alg:pgd} converge to the optimal solution of
\cref{eq:ucopt}, denoted $x^\star$, which is a stationary point.
\begin{theorem}[Fixed Point Property of PGD \cite{parikh2014}]\label{thm:fp}
A solution $x^\star$ is the optimal solution to \cref{eq:ucopt} if and
only if 
\begin{equation}
x^\star = \text{\bf prox}_{\lambda g} \left(x^\star - A^*(Ax^\star - b)\right)
\label{eq:pgdfp}
\end{equation}
\end{theorem}

Assuming $\text{\bf prox}_{\lambda g}$ has a closed form or easy-to-implement
solution, PGD is a robust and simple-to-implement algorithm that achieves
$O(1/k)$ iterative convergence \cite{beck2009}.
To improve the iterative convergence, FISTA \cite{beck2009} was proposed
that achieves $O(1/k^2)$ iterative convergence by smartly choosing different
point $z_k$ (instead of $x_k$) to evaluate \cref{eq:pgdupdate}.
Listed as \cref{alg:fista} is the FISTA algorithm that utilizes
$\beta_k$ proposed in \cite{chambolle2015} to solve \cref{eq:ucopt}.
\begin{algorithm}
\caption{FISTA\cite{beck2009} with $\beta_k$ from \cite{chambolle2015}}
\label{alg:fista}
\begin{algorithmic}
\STATE{\textit{Inputs:}\begin{itemize}
  \item[-] forward model $A$
  \item[-] measurements $b$
  \item[-] proximal operator $\text{\bf prox}_{\lambda g}$
  \item[-] $x_{-1} = x_0 = 0$
  \end{itemize}}
\STATE{\textit{Step k:} $(k \geq 0)$ Compute
\begin{subequations}\begin{align}
x_{k+1} &= \text{\bf prox}_{\lambda g} \left(z_k - A^*(Az_k - b)\right)\\
\beta_k &= k/(k + 3) \\
z_{k+1} &= x_k + \beta_k \left(x_k - x_{k-1}\right)
\end{align}\end{subequations}}
\end{algorithmic}
\end{algorithm}

A useful property of proximal operators that will be utilized in the sequel is
the \textit{firm non-expansiveness} property.
\begin{lemma}[Firm Non-Expansiveness Property \cite{parikh2014}]
\label{lemma:fne}
For all $x, y \in \mathbb{C}^n$, 
\begin{equation}
\lVert \text{\bf prox}_{\lambda g}(x) - \text{\bf prox}_{\lambda g}(y) \rVert_2
\leq \lVert x - y \rVert_2
\end{equation}
\end{lemma}

\subsection{Slow Convergence of PGD}
\Cref{lemma:fne} helps illustrate how the convergence of iterates of PGD is
affected by the eigenvalue spectrum of $A$.
\begin{theorem}\label{thm:pgdbnd}
The error of iterates $x_k$ in \cref{alg:pgd} with respect to $x^\star$ are
upper-bounded by the spectrum of $(I - A^*A)$.
\end{theorem}
\begin{proof}
Let $e_k = x_k - x^\star$ be the error over iterations.
Subtracting \cref{eq:pgdupdate} with \cref{eq:pgdfp}, taking the $l_2-$norm
and utilizing \Cref{lemma:fne} yields:
\begin{equation}
\lVert e_{k+1} \rVert_2 \leq \lVert (I - A^*A) e_k \rVert_2
\label{eq:pgdbnd}
\end{equation}
In particular, splitting $e_k$ into $s_k + t_k$ where $s_k \in\text{null}(A)$
and $t_k \in \text{null}(A)^\perp$ results in:
\begin{equation}
\lVert e_{k+1} \rVert_2^2 \leq
\lVert (I - A^*A) t_k \rVert_2^2 + 
\lVert s_k \rVert_2^2
\label{eq:pgdcnv}
\end{equation}
\end{proof}

Therefore, the decrease in error of $e_k$ within $\text{null}(A)^\perp$
(i.e. $t_k$) is upper-bounded by the eigenvalue spectrum of $I - A^*A$.
This bound is strict when $A^*A$ is injective.
\begin{corollary}
If $A^*A$ has no zero valued eigenvalues, the inequality in \cref{eq:pgdbnd} is
strict.
\end{corollary}

Thus, inverse problems involving an $A^*A$ with small eigenvalues
may suffer from slow iterative convergence.

\subsection{Main Results\label{sec:main}}
\Cref{thm:pgdbnd} motivates the design of a preconditioner that minimizes
the magnitude of the eigenvalues of $I - A^*A$ for faster convergence.
The main result of this work is to modify the gradient update steps in
\cref{alg:pgd} and \cref{alg:fista} using preconditioner $P = p(A^*A)$, where
$p$ is designed in such a way that $p(A^*A)$ increases the contributions of the
smaller eigenvalues of $A^*A$ over iterations for faster convergence.

Let $p$ have fixed degree $d$, where $d$ is a hyper-parameter that is to be
tuned for the application of interest.
The coefficients of $p$ are calculated by optimizing over \cref{eq:l2cost},
which is motivated in \Cref{sec:poly}.
\begin{equation}
p = \left\{\underset{q}{\text{argmin}} \int_{z=0}^1 (1 - q(z) z)^2 dz\right.
\label{eq:l2cost}
\end{equation}

The polynomial from \cref{eq:l2cost} is then used to derive preconditioner
$p(A^*A)$ that is included in \cref{alg:pgd} to arrive at \cref{alg:ppgd}.
The derivation is in \Cref{sec:ppgddrv}.
\begin{algorithm}
\caption{Polynomial Preconditioning for PGD}
\label{alg:ppgd}
\begin{algorithmic}
\STATE{\textit{Inputs:}
  \begin{itemize}
  \item[-] forward model $A$
  \item[-] measurements $b$
  \item[-] proximal operator $\text{\bf prox}_{\lambda g}$
  \item[-] optimized polynomial $p$ from \cref{eq:l2cost}
  \item[-] $y_0 = 0$.
  \end{itemize}}
\STATE{\textit{Step k:} $(k \geq 0)$ Compute
\begin{equation}\label{eq:ppgdupdate}
y_{k+1} = \text{\bf prox}_{\lambda g} \left(y_k - p(A^*A) A^*(Ay_k - b)\right)
\end{equation}}
\end{algorithmic}
\end{algorithm}
\begin{remark}
The iteration variable name has been changed in \cref{alg:ppgd} (and
the following \cref{alg:pfista}) from $x_k$ to $y_k$ to emphasize that the
preconditioner is being used.
\end{remark}

Similarly, the preconditioner is integrated into \cref{alg:fista} to arrive
at \cref{alg:pfista}.
\begin{algorithm}
\caption{Polynomial Preconditioning for FISTA}
\label{alg:pfista}
\begin{algorithmic}
\STATE{\textit{Inputs:}
  \begin{itemize}
  \item[-] forward model $A$
  \item[-] measurements $b$
  \item[-] proximal operator $\text{\bf prox}_{\lambda g}$
  \item[-] optimized polynomial $p$ from \cref{eq:l2cost}
  \item[-] $y_{-1} = y_0 = 0$.
  \end{itemize}}
\STATE{\textit{Step k:} $(k \geq 0)$ Compute
\begin{subequations}\begin{align}
y_{k+1} &= \text{\bf prox}_{\lambda g} \left(z_k - p(A^*A)
A^*(Az_k - b)\right)\\
\beta_k &= k/(k + 3) \\
z_{k+1} &= y_k + \beta_k \left(y_k - y_{k-1}\right)
\end{align}\end{subequations}}
\end{algorithmic}
\end{algorithm}

\subsection{Deriving Polynomial Preconditioning for PGD}\label{sec:ppgddrv}
It helps to derive \cref{alg:ppgd} before discussing \cref{eq:l2cost}.
To motivate \cref{alg:ppgd},  first consider the exact formulation \cref{eq:cs}
with $b\in\text{range}(A)$.
Let the singular value decomposition (SVD) of $A$ in dyadic form be as follows:
\begin{equation}
A(\cdot) = \sum_{i=1}^j \sigma_i \langle \cdot, v_i \rangle u_i
\label{eq:svd1}\end{equation}
Here, $j \leq n, 1\geq\sigma_1\geq\sigma_2\geq\dots\geq\sigma_j>0$ are the
singular values of $A$, $\{u_i\}$ and $\{v_i\}$ are the left and right
singular vectors respectively, and
Let $p(z)$ be a polynomial of degree $d$ such that $p(z) > 0$ for $z\in(0,1]$,
and let $P = p(AA^*)$.
\begin{equation}
P(\cdot) = \sum_{i=1}^j p(\sigma_i^2) \langle \cdot, u_i \rangle u_i
\label{eq:psvd}
\end{equation}
As an ansatz, let $P^{\frac{1}{2}}$ be the square root of $P$:
\begin{equation}
P^{\frac{1}{2}}(\cdot) = \sum_{i=1}^j
  \left[p(\sigma_i^2)\right]^{\frac{1}{2}}
  \langle \cdot, u_i \rangle u_i
\end{equation}

\begin{lemma}\label{lemma:pcs}
The condition $Ax = b$ in \cref{eq:cs} is equivalently enforced by the
constraint $P^{\frac{1}{2}} A x = P^{\frac{1}{2}}b$.
\end{lemma}
\begin{proof}
The condition $p(z) > 0$ for $z \in (0, 1]$ implies $P$
and $P^{\frac{1}{2}}$ are injective when the domain and co-domain for
both operators are restricted to $\text{range}(A)$.
\end{proof}

\Cref{lemma:pcs} motivates the following ``preconditioned'' formulation
that has the same solution as \cref{eq:cs}.
\begin{equation}
x^\star = \left\{
\begin{aligned}
& \underset{x}{\text{argmin}}
& & g(x) \\
& \text{subject to}
& & P^{\frac{1}{2}}Ax = P^{\frac{1}{2}} b
\end{aligned}\right.
\label{eq:pcs}
\end{equation}

To account for model and measurement errors, \cref{eq:pcs} is relaxed to the
following constrained formulation:
\begin{equation}
x^\star = \left\{
\begin{aligned}
& \underset{x}{\text{argmin}}
& & g(x) \\
& \text{subject to}
& & \lVert P^{\frac{1}{2}}(Ax - b) \rVert_2 \leq \epsilon_p
\end{aligned}\right.
\label{eq:pcopt}
\end{equation}
Note that $\epsilon$ and $\epsilon_p$ in the constraints of
\cref{eq:copt} and \cref{eq:pcopt} respectively are likely to be different
as $P^{\frac{1}{2}}$ is not necessarily unitary.

Similarly to how \cref{eq:copt} is recast to \cref{eq:ucopt},
\cref{eq:pcopt} is solved by the following for an appropriate
choice of $\lambda_p$.
\begin{equation}
x^\star = \left\{
\begin{aligned}
& \underset{x}{\text{argmin}}
& & \frac{1}{2} \lVert P^{\frac{1}{2}}(Ax - b) \rVert_2^2 + \lambda_p g(x)
\end{aligned}\right.
\label{eq:pucopt}
\end{equation}

Defining $A_P = P^{\frac{1}{2}}A$ and $b_P = P^{\frac{1}{2}}b$, the PGD
iterations to solve \cref{eq:pucopt} and the fixed point condition
(\cref{thm:fp}) of \cref{eq:pucopt} is as follows:
\begin{subequations}\begin{align}
y_{k+1} &= \text{prox}_{\lambda_p g}(y_k - A_P^*(A_Py_k - b_P)) \\
y^\star &= \text{prox}_{\lambda_p g}(y^\star - A_P^*(A_Py^\star - b_P))
\end{align}\label{eq:ppgd_step}\end{subequations}
Here, $y^\star$ is the solution to \cref{eq:pucopt}.

To significantly simply the iterations of \cref{eq:ppgd_step}, the
permutability of spectral functions and $A^*A$, listed as
\cref{lemma:permute}, is leveraged.

\begin{lemma}[Permutability of Spectral Functions]\label{lemma:permute}
\begin{subequations}\begin{align}
A^* p(AA^*) A &= p(A^*A) A^*A \\
A^* p(AA^*)   &= p(A^*A) A^*
\end{align}
\end{subequations}
\end{lemma}
\begin{proof}
Since $p$ is a spectral function of $A$, $p(AA^*)$ and $p(A^*A)$ can be
evaluated with respect to the singular value decomposition of $A$ as defined in
\cref{eq:svd1} and \cref{eq:psvd} to verify this result.
\end{proof}

Incorporating \cref{lemma:permute} into \cref{eq:ppgd_step} results in:
\begin{subequations}\begin{align}
y_{k+1} &= \text{prox}_{\lambda_p g}(y_k - p(A^*A)A^*(Ay_k - b)) \\
y^* &= \text{prox}_{\lambda_p g}(y^* - p(A^*A)A^*(Ay^* - b))
\end{align}\label{eq:ppgd}\end{subequations}

This yields the \cref{alg:ppgd}.

\begin{remark}
Solving \cref{eq:pucopt} with FISTA and using \cref{lemma:permute} results
in \cref{alg:pfista}.
\end{remark}

\subsection{Polynomial Design}\label{sec:poly}
A line of reasoning similar to \cref{thm:pgdbnd} is used to arrive at
\cref{eq:l2cost}.

\begin{theorem}\label{thm:ppgdbnd}
The error of iterates $y_k$ in \cref{alg:ppgd} with respect to $y^\star$ are
upper-bounded by the spectrum of $(I - p\,(A^*A)\;A^*A)$.
\end{theorem}
\begin{proof}
Let $e_k = y_k - y^\star$ be the error over iterations.
Subtracting the equations of \cref{eq:ppgd}, taking the $l_2-$norm and
utilizing \Cref{lemma:fne} yields:
\begin{equation}
\lVert e_{k+1} \rVert_2 \leq \lVert (I - p\,(A^*A) A^*A) e_k \rVert_2
\label{eq:ppgdbnd}
\end{equation}
Similarly to \cref{thm:pgdbnd}, splitting $e_k$ into $s_k + t_k$ where
$s_k \in\text{null}(A)$ and $t_k \in \text{null}(A)^\perp$ results in:
\begin{equation}
\lVert e_{k+1} \rVert_2^2 \leq
\lVert (I - p\,(A^*A) A^*A) t_k \rVert_2^2 + 
\lVert s_k \rVert_2^2
\label{eq:ppgdcnv}
\end{equation}
\end{proof}

\Cref{thm:ppgdbnd} motivates finding a polynomial $p$ such that
$I - p\,(A^*A) A^*A$ is as close to zero as possible.
As the dimensions of $A$ are typically very large, it is not
computationally feasible (in terms of processing time) to perform an
eigenvalue decomposition of $A^*A$ to use as prior information for
polynomial design.
To avoid this, the coefficients of the polynomial $p$ is found by
optimizing the continuous approximation of the induced norm.
This is listed as \cref{eq:l2cost}.
The polynomial that minimizes \cref{eq:l2cost} in-turn minimizes
the induced $l_2$-norm of $I - p(A^*A) A^*A$ in \cref{thm:ppgdbnd}.

In other words, the component $t_k$ in $v_i$ (as in \cref{eq:ppgdcnv}
and \cref{eq:svd1}) is upper-bounded by $|1 - p(\sigma_i^2) \sigma_i^2|$, 
which according to \cref{eq:l2cost}, is minimized to be close to zero.
Note that the larger degree $d$ of polynomial $p$, the better
$(1 - p(x) x)^2$ (equivalently $|1 - p(x) x|$), can approximate the zero
function.

Priors on the spectrum can be easily incorporated into \cref{eq:l2cost} as
follows:
\begin{equation}
p = \left\{\underset{q}{\text{argmin}} \int_{z=0}^1 w(z) (1 - q(z) z)^2
dz\right.
\label{eq:wl2cost}
\end{equation}
Here, $w$ can weight the cost to prioritize certain components of the
spectrum.

\begin{remark}
At first glance, minimizing the following objective instead of
\cref{eq:l2cost} is preferable as it directly translates into minimizing
the appropriate induced norm of \cref{thm:ppgdbnd}:
\begin{equation}\label{eq:linfcost}
p = \left\{\underset{q}{\text{argmin}}\;\underset{z \in (0, 1]}{\max} \;
|1 - q(z) z|\right.
\end{equation}
It is well known that Chebyshev polynomials of the first kind can be
used to derive the optimal polynomials such that the maximum absolute value
of that polynomial over a specified interval is minimized.
However, defining $r(z) = 1 - q(z) z$, and using Chebyshev polynomials
to determine $r$ yields a polynomial with $r(z) = 1$ for multiple values
of $z \in [0, 1]$ due to the constraint $r(0) = 1$, which implies
the components of $t_k$ in \cref{eq:ppgdcnv} corresponding to eigenvalues
$\sigma^2$ such that $r(\sigma^2) = 1$ will not decrease.
However, if the minimum non-zero eigenvalue $\mu = \sigma_j^2$ of $A^*A$
is known a-priori, the following polynomial minimizes \cref{eq:linfcost}
over the interval $[\mu, 1]$ (See
\cite{shewchuk1994, kaniel1966, van1986, johnson1983}):
\begin{subequations}\begin{align}
r(z)& = \frac{T_{d+1}\left(\frac{1 + \mu - 2z}{1 - \mu}\right)}
              {T_{d+1}\left(\frac{1 + \mu}{1 - \mu}\right)} \\
p(z)& = \frac{1 - r(z)}{z}
\end{align}\end{subequations}
Here, $T_{d + 1}$ is the Chebyshev polynomial of the first kind of
degree $d + 1$.
In practice, \cref{eq:l2cost} is preferred as it is often computationally
expensive to estimate $\mu$ unless $A^*A$ happens to be injective,
in which case $\mu$ can be estimated by performing power-iteration on
$I - A^*A$.
\end{remark}

\begin{remark}
Rather interestingly, while \cref{eq:l2cost} and \cref{eq:linfcost}
were motivated by studying the evolution of iterates in PGD when solving
regularized linear inverse problems in the form of \cref{eq:ucopt}, the
exact formulation for polynomial optimization was studied in
\cite{johnson1983} as a means of accelerating the convergence of
Conjugate Gradient for ordinary least squares optimization, where
\cref{eq:l2cost} and \cref{eq:linfcost} were optimized to construct an
incomplete inverse of $(A^*A)$ to use a preconditioner.
\end{remark}

\Cref{thm:ppos} verifies that the positivity assumption in \cref{lemma:pcs}
is satisfied by optimized results of \cref{eq:l2cost} and \cref{eq:linfcost}.
\begin{theorem}[Polynomial Positivity]\label{thm:ppos}
The polynomial $p$ derived from optimized \cref{eq:l2cost} and
\cref{eq:linfcost} satisfies the constraint $p(z) > 0$ for $z \in (0, 1]$.
\end{theorem}
\begin{proof}
This is a consequence of Theorem 4 of \cite{johnson1983}.
\end{proof}

\subsection{Error Bound\label{sec:err}}

With the simplifying assumptions that $A^*A$ is injective and
$\lambda_p = \lambda$, \cref{thm:err} bounds the difference between
the solutions of \cref{eq:ucopt} and \cref{eq:pucopt}.
Let $\lVert \cdot \rVert_{2 \rightarrow 2}$ denote the induced $l_2-$norm.

\begin{theorem}\label{thm:err}
Assume $A^*A$ is injective, and $\lambda = \lambda_p$ in \cref{eq:ucopt}
and \cref{eq:pucopt} respectively.
Let $e_k, \gamma$ and $\delta$ be defined as:
\begin{subequations}\begin{align}
e_k &= y_k - x^\star \\ 
\gamma &= \lVert I - p(A^*A)A^*A\rVert_{2 \rightarrow 2} \\
\delta &= \lVert (I-p(A^*A))A^*(Ax^\star - b) \rVert_2
\end{align}\end{subequations}
Here, $x^\star$ is the optimal solution to \cref{eq:ucopt}, $y_k$ are the
iterates of \cref{alg:ppgd} that leverages polynomial $p$.
Then, the limit of the difference, $e_\infty$, satisfies the following:
\begin{equation}
\lVert e_\infty \rVert_2 \leq \frac{\delta}{1 - \gamma}
\label{eq:errbnd}\end{equation}
\end{theorem}
\begin{proof}
Subtracting \cref{eq:ppgdupdate} from \cref{eq:pgdfp}, taking the $l_2-$norm of
the difference and utilizing \cref{lemma:fne} with the triangle inequality
yields:
\begin{equation}
\lVert e_{k+1}\rVert_2 \leq \lVert (I - PA^*A)e_k \rVert_2 +
                            \lVert (I-P)A^*(Ax^\star - b) \rVert_2
\end{equation}
This reduces to:
\begin{equation}
\lVert e_{k+1}\rVert_2 \leq \gamma \lVert e_k \rVert_2 + \delta
\label{eq:series}\end{equation}
As $A^*A$ is injective, it follows that, after optimizing for $P$
via \cref{eq:l2cost} or \cref{eq:linfcost}, $\gamma < 1$.
This in-turn implies \cref{eq:series} converges and the limit
$e_\infty$ satisfies \cref{eq:errbnd}.
\end{proof}

For most regularized linear inverse problems, $\lVert Ax^\star - b\rVert$ is
small, which in turn implies the error between \cref{eq:ucopt} and
\cref{eq:pucopt} is also small as:
\begin{equation}
\delta \leq \lVert I - p\,(A^*A)\rVert_{2\rightarrow 2} \;
\lVert A^* (Ax^\star - b) \rVert_2
\end{equation}

\subsection{Implementation Details and Complexity Analysis}\label{sec:complex}

Prior works \cite{johnson1983, bhotto2015} that used polynomial preconditioning
utilized the matrix entries of $A$ to explicitly pre-calculate $p(A^*A)$ to
save computation time.
Since this work aims to leverage matrix-free implementations of $A$, this is no
longer possible.
Rather, utilizing a polynomial preconditioner of degree $d$ involves $(d + 1)$
$A^*A$ evaluations per iteration, thereby increasing the per-iteration cost
compared to PGD and FISTA.
However, the main benefit of the preconditioner is that, for the same number of
$A^*A$ evaluations, the components of the iterates corresponding to the smaller
eigenvalues of $A^*A$ are more explicitly targeted.
To make sure this point is well reflected in the sequel, the experiments in
\Cref{sec:experiments} reports convergence as a function of the total number of
$A^*A$ evaluations (and the number of iterations, which reflects the number of
proximal operator evaluations) in addition to reporting the iterative
convergence as a function of wall time to demonstrate real world performance.

To give an example, \cref{fig:complexity} depicts the iteration progress
assuming a degree $3$ polynomial optimized using \cref{eq:l2cost} after
$4\;A^*A$ evaluations for \cref{alg:pgd} and \cref{alg:ppgd}, which shows that
the polynomial preconditioner aids the convergence of components corresponding
to the lower eigenvalues of $A^*A$ at the cost of slightly slower convergence
of the higher eigenvalues.
In the sequel, the degree of the polynomial is a hyper-parameter that is tuned
for the application of interest.

\begin{figure}[ht]
  \centering
  \includegraphics[width=\textwidth]{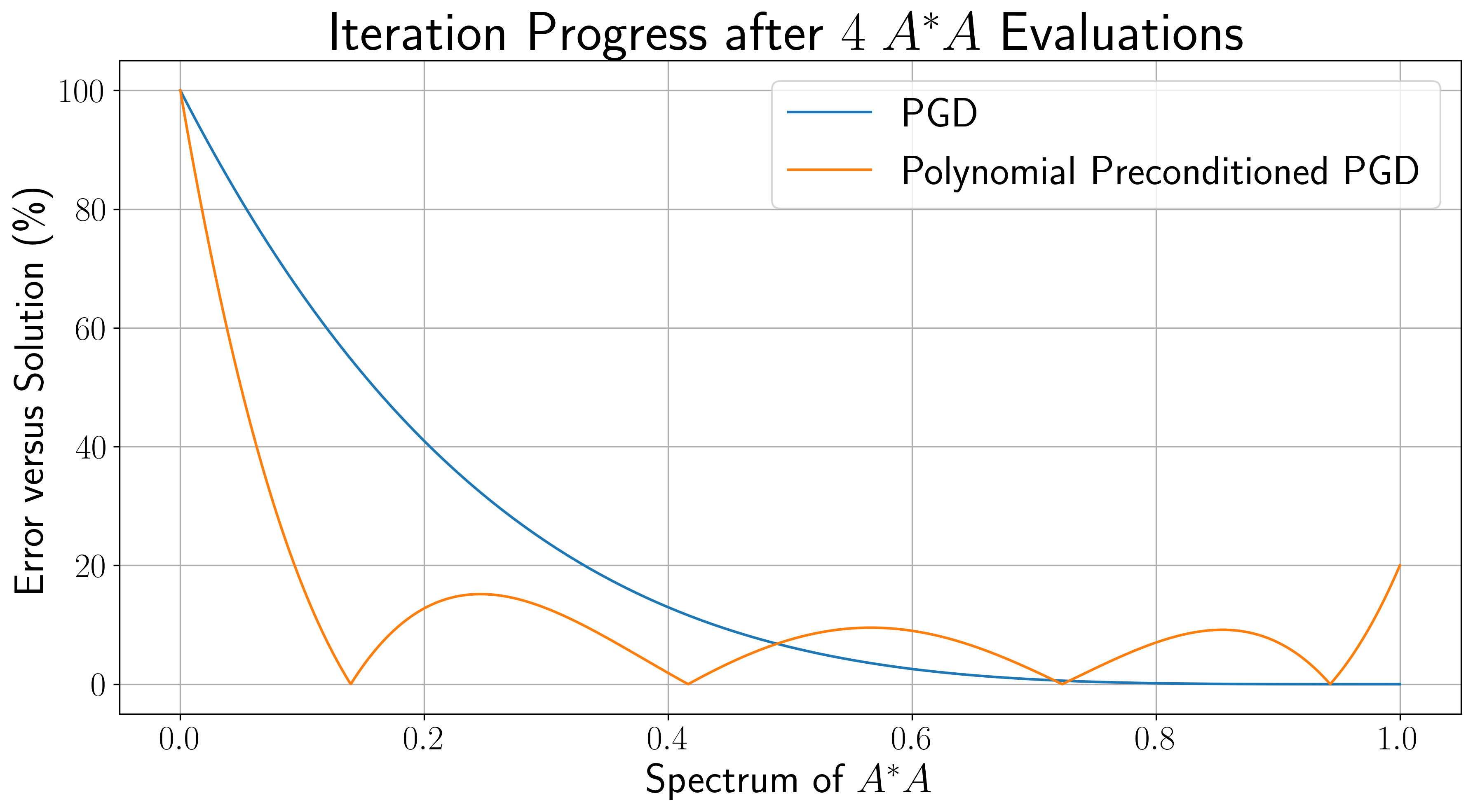}
  \caption{
    Iteration progress of \cref{alg:pgd} and \cref{alg:ppgd} assuming a degree
    $3$ polynomial optimized using \cref{eq:l2cost} after $4 A^*A$ evaluations.
    The x-axis denotes the spectrum of $A^*A$.
    The y-axis denotes the percentage of decrease of error with respect to
    the respective final iterates.
    This figure shows that the polynomial preconditioner aids the convergence
    of components corresponding to the lower eigenvalues of $A^*A$ at the cost
    of slower convergence of the higher eigenvalues.
  }
  \label{fig:complexity}
\end{figure}

\begin{remark}
Note that for both, numerical stability and achieving $(d + 1)$ evaluations of
$A^*A$ per iteration, it is important to leverage the polynomial structure when
evaluating the preconditioner.
For example, consider pseudo-code implementation in Listing \ref{pcode} that
utilizes recursion.
It is assumed that ``ListOfPolynomialCoefficients'' returns the coefficients of
$p$ in the order:
\begin{equation}p(x) = \sum_{i=0}^d c[i] x^i \end{equation}
\end{remark}

\begin{lstlisting}[float, language=Python, frame=single, label=pcode,
                   caption=Pseudo-code for Polynomial Evaluation]
I = IdentityOperator()
N = ForwardModelNormalOperator()

def CreatePolynomialPreconditioner(c):
  if len(c) == 1:
    return c[0] * I
  else:
    return c[0] * I + N * CreatePolynomialPreconditioner(c[1:])

coeffs = ListOfPolynomialCoefficients()
Preconditioner = CreatePolynomialPreconditioner(coeffs)
\end{lstlisting}

\subsection{Interpretation and Noise Coloring}

At first glance, the proposed preconditioner is similar to the
left-preconditioning methods such as \cref{eq:ducopt}.
However, instead of using a diagonal matrix to weight to the entries
(or coordinates) of the measurement $b$, $P^{1/2}$ ``spectrally'' weights $b$.
In other words, $P^{1/2}$ is a diagonal matrix with respect to the left
singular vectors of $A$ (or $\{u_i\}$).
A unique consequence of this property is that, in the absence of
regularization, polynomial preconditioning result in the same solution as
\cref{eq:ucopt} and thus introduces no noise coloring.
This is not necessarily true for general left-preconditioning methods.

\begin{theorem}
When $\lambda = \lambda_p = 0$ in \cref{eq:ucopt} and \cref{eq:pucopt}
respectively, the solutions to \cref{eq:ucopt} and \cref{eq:pucopt} are
identical as long as both optimizations are equivalently initialized.
\end{theorem}
\begin{proof} Using \cref{eq:svd1} and \cref{eq:psvd},
\begin{equation}\begin{aligned}
x \text{ solves } \cref{eq:ucopt}
  &\iff& \sigma_i \langle x, v_i \rangle = \langle b, u_i\rangle\\
  &\iff& p(\sigma_i^2)^\frac{1}{2} \sigma_i \langle x, v_i \rangle = 
         p(\sigma_i^2)^\frac{1}{2} \langle b, u_i\rangle\\
  &\iff& x \text{ solves } \cref{eq:pucopt}
\end{aligned}\label{eq:samel2}\end{equation}
\end{proof}

\section{MRI Experiments}\label{sec:experiments}

In the spirit of reproducible research, the data and code used to
perform the following experiments can be found at:
\begin{center}
\url{https://github.com/sidward/ppcs}%
\footnote{\url{https://doi.org/10.5281/zenodo.6475880}}
\end{center}
All reconstructions were implemented in the Python programming
language using
SigPy\footnote{\url{https://doi.org/10.5281/zenodo.5893788}} \cite{ong2019}.
The polynomial optimizations \cref{eq:l2cost} were performed using
SymPy\footnote{\url{https://doi.org/10.7717/peerj-cs.103}}
\cite{sympy}, with the latter leveraging an excellent Chebyshev polynomial
package available at \url{https://github.com/mlazaric/Chebyshev}%
\footnote{\url{https://doi.org/10.5281/zenodo.5831845}}.

To verify the efficacy of the preconditioner, three varied MRI reconstructions
were studied using the unconstrained formulations \cref{eq:ucopt} and
\cref{eq:pucopt}.
The Nesterov accelerated variants, \cref{alg:fista} and \cref{alg:pfista},
were used over their respective non-accelerated counterparts.
For all cases, the corresponding measurement matrix $A$ was normalized to have
a unitary induced $l_2$-norm, and the measurement vector $b$ was normalized to
have unitary $l_2$-norm.
All experiments were performed on an Intel (R) Xeon Gold 5320 CPU and
NVIDIA(R) RTX A6000 GPUs.

\subsection{Parameter Selection}\label{sec:param}

Before discussing the specific experiments, this section will describe the
efforts undertaken to ensure that each algorithm is portrayed in the best
light.

Solving linear inverse problems posed as \cref{eq:ucopt} and \cref{eq:pucopt}
using FISTA, IFISTA, ADMM and the proposed method requires the user to specify
the regularization value (for each algorithm) and the desired number of $A^*A$
evaluations per iteration (for each algorithm modulo FISTA).
FISTA utilizes one $A^*A$ evaluation per iteration.
IFISTA and the proposed polynomial preconditioner, when using a polynomial $d$,
utilizes $(d + 1)\;A^*A$ evaluations per iteration.
ADMM utilizing $n\;A^*A$ evaluations per iteration implies $(n - 1)$ conjugate
gradient (CG) iterations during the inner loop of ADMM (as one $A^*A$
evaluation is needed to calculate the initial residual term in the CG
algorithm).
IFISTA and the proposed method require additional tuning of the regularization
value as argued by the introduction of $\lambda_p$ in \cref{eq:pucopt}.

The convergence of the iterates of the respective algorithms and the quality of
the converged result are dependent on the regularization value chosen and the
number of $A^*A$ evaluations per iteration.
A fair comparison between the algorithms necessitates a search through the
large parameter space of regularization values and the number of normal
evaluations.
It is difficult to manually choose parameters based on the observed
convergence and reconstruction quality.
To overcome this, the following programmatic procedure is used to select the
parameters.
\begin{enumerate}
\item
  For all algorithms, the maximum number of total $A^*A$ evaluations is fixed.
\item
  The minimum normalized root mean squared error (NRMSE) achieved by FISTA
  (with respect to the reference image) among all tested regularization values
  $(\lambda)$ is noted as $\epsilon_f\,\%$.
\item
  For each method, if a parameter set yields a reconstruction with NRMSE (with
  respect to the reference image) that is greater than $(\epsilon_f + 2)\%$,
  then that parameter set is discarded.
  The ``$+2$'' term was set to account for the finite resolution of the grid
  search of regularization values $(\lambda)$ and potential numerical errors.
  This ``$+2$'' term was observed to yield qualitatively similar reconstruction
  to the minimum NRMSE FISTA result for each of the following experiments.
  Note that the reconstruction error achieved might be lower than
  $\epsilon_f\%$.
\item 
  For each method, the convergence curves (with respect to the last iterate) as
  a function of the observed wall-time is calculated for the remaining
  parameter sets that satisfy the $(\epsilon_f + 2)\%$ constraint.
\item
  A log-linear fitting of the convergence curves is performed, and the
  parameters with the most negative slope (after the log-linear fitting) is
  chosen to represent the method.
\end{enumerate}

By selecting parameters that yield reconstructions that fall within
$(\epsilon_f + 2)\%$ NRMSE of the reference, the reconstruction quality is
assured.
Given the subset of parameters, the convergence is calculated with respect to
the last iterate of the respective methods instead of the reference image to
directly show-case iterative convergence, which allows the resulting plots
in the following experiments to be directly interpreted in the context of
\cref{eq:pgdcnv,eq:ppgdcnv}.

Please see \Cref{sec:exp1} for a visual example of the parameter selection
process proposed in this section.

\subsection{Cartesian MRI}\label{sec:exp1}

The first experiment is a 2D-Cartesian Compressed Sensing Knee application
using publicly available data \cite{ong2018, sawyer2013}.
The reference data was acquired using a 3D-FSE CUBE acquisition with proton
density weighting that included fat saturation\cite{sawyer2013} on a 3T
whole-body scanner (GE Healthcare, Waukesha, WI) using an 8-channel HD knee
coil (GE Healthcare, Milwaukee, WI,USA) with an echo time (TE) of 25ms and
repetition time (TR) of 1550ms.
The field-of-view was 160mm, the matrix size was $320\times320$, slice
thickness was 0.6mm and 256 slices were acquired.
The reference data was fully-sampled and satisfied the Nyquist criterion.
This reference data was retrospectively under-sampled by approximately $R=7.21$
times using a variable density Poisson disc sampling mask generated by
BART\cite{uecker2015}.

The unconstrained reconstruction formulation \cref{eq:ucopt} for this
experiment is as follows:
\begin{equation}\label{eq:exp1}
\begin{aligned}
x^\star &=& \left\{\underset{x}{\text{argmin}} \; \frac{1}{2}
        \lVert MFS x - b \rVert_2^2 + \lambda \lVert W x\rVert_1\right.
\end{aligned}
\end{equation}
Here, $W$ is the forward Daubechies-4 Wavelet transform, $S$ is the SENSE
model of the parallel-imaging acquisition \cite{pruessmann1999} estimated
using \cite{uecker2014}, $F$ is the 2D-Fourier transform and $M$ is the
Poisson disc sampling mask.

This application was solved with FISTA, IFISTA, ADMM and FISTA with the
polynomial preconditioner.
Each algorithm was allowed to run for a maximum of 60 $A^*A$ evaluations.
For FISTA, IFISTA and polynomial preconditioned FISTA, regularization values
$\lambda$ were varied as:
\begin{equation}
\lambda \in \left\{\frac{10^{-2}}{1.5^k} : k = 0, 1, \dots, 14 \right\}
\label{eq:lamda}
\end{equation}
For IFISTA, ADMM and the polynomial preconditioned FISTA, the number of
$A^*A$ evaluations were varied as:
\begin{equation}
A^*A \text{ evaluations per iteration} \in \left\{2, 3, 4, 5, 6\right\}
\label{eq:aha}
\end{equation}
The number of $A^*A$ evaluations per iteration were chosen to cleanly divide
the maximum number of $A^*A$ evaluations for fair comparison between methods.
Lastly, the regularization value from the FISTA experiments with the least
NRMSE against the reference image was used for ADMM, with the step size
$(\rho)$ varied as:
\begin{equation}
\rho \in \left\{ 3^k : k = -7, -6, \dots, -1, 0, 1, \dots, 6, 7\right\}
\label{eq:rho}
\end{equation}

\begin{figure}[ht]
  \centering
  \includegraphics[width=0.75\textwidth]{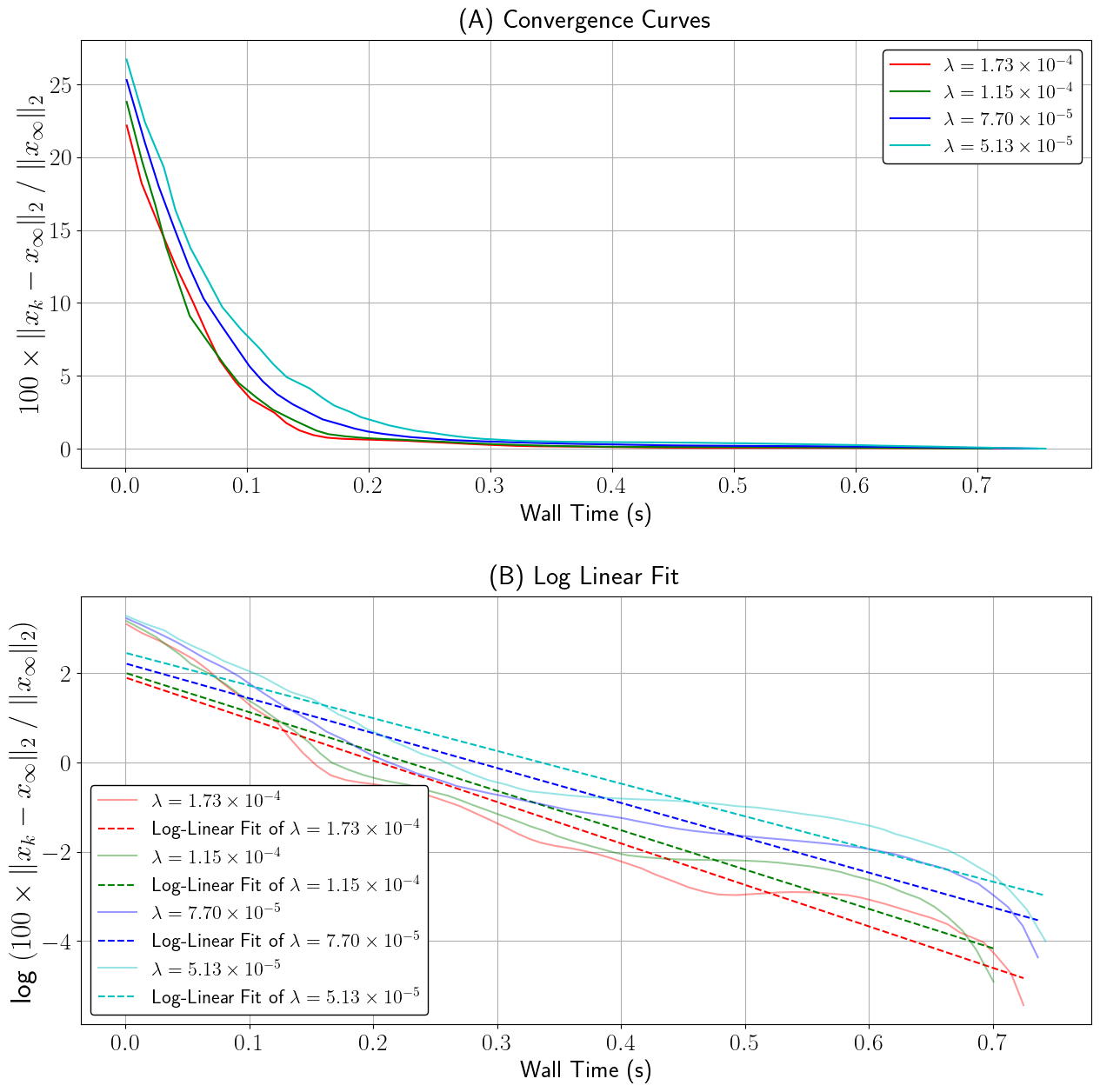}
  \caption{
    \textbf{Parameter Selection.} This figure demonstrates how the
    hyper-parameters were chosen to represent FISTA.
    Among the hyper-parameters tested in \cref{eq:lamda}, the minimum NRMSE
    against the reference image achieved by FISTA was $\epsilon_f = 13.61\%$.
    The four $\lambda$ values shown in (A) and (B) achieved
    $(\epsilon_f + 2)\%$ NRMSE.
    In order to programmatically pick the fastest converging value between the
    listed four values, a log-linear fitting of the convergence curves is
    performed, as shown in (B).
    The parameter with the most negative log-linear slope is chosen to
    represent the algorithm, which in this case is the
    $\lambda=1.73 \times 10^{-4}$ curve.
  }
  \label{fig:logfit}
\end{figure}
To give a visual example of how the programmatic selection works
(\Cref{sec:param}), among the $\lambda$ values tested in \cref{eq:lamda},
the minimum NRMSE achieved by FISTA was $\epsilon_f=13.61\%$.
It is observed that for FISTA, four different regularization values $\lambda$
achieved $(\epsilon_f + 2)\%$ or less NRMSE with respect to the reference
image.
Qualitatively, the $\lambda=1.73 \times 10^{-4}$ curve achieved the fastest
iterative convergence as shown in \Cref{fig:logfit}(A).
This is seen to translate to the most negative slope for the corresponding
log-linear fit as shown in \Cref{fig:logfit}(B) and is thus the parametric
value chosen to represent FISTA.

For IFISTA, ADMM and the proposed method, the identical process was used (with
the same $(\epsilon_f + 2)\%$ NRMSE condition) to choose representative
regularization value and number of $A^*A$ evaluations per iteration.

The convergence curves of the respective methods after the programmatically
chosen parameters with respect to both, the number of $A^*A$ evaluations and
observed wall times, are depicted in \Cref{fig:ccs_conv}, and the final
iterates of the corresponding methods are depicted in \Cref{fig:ccs_recons}
along with the sampling mask $M$ leveraged in \cref{eq:exp1}.

\begin{figure}[ht]
  \centering
  \includegraphics[width=\textwidth]{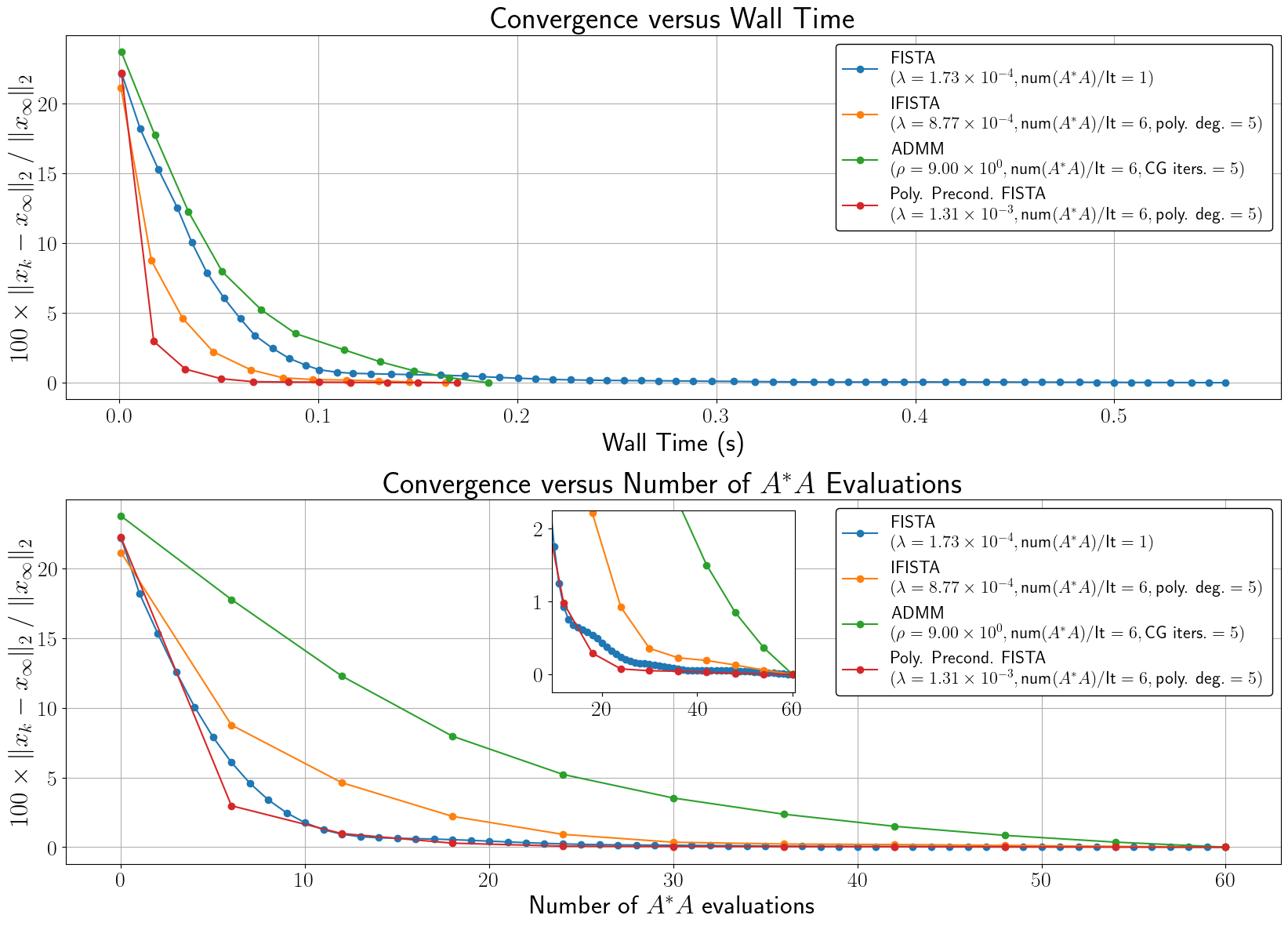}
  \caption{\textbf{Cartesian MRI Convergence Results.}
    This figure depicts the convergence results of the respective methods
    described in \Cref{sec:exp1} given the hyper-parameters selected as in
    \Cref{sec:param} with $\epsilon_f=13.61\%$.
    Given the chosen hyper-parameters, the error over iterations with respect
    to the last iteration of each respective method is plotted.
    The $k^{th}$ iteration and last iteration are labelled as $x_k$ and
    $x_\infty$ respectively.
    The legend on the top right depicts the algorithm and the chosen
    hyper-parameters.
    The x-axis of the top and bottom subplots denotes the total number of
    $A^*A$ evaluations and measured wall-times respectively.
    The circular-markings on each line denote the respective iteration points.
    The number of proximal operators evaluated by a point on the x-axis is
    equal to the number of iterations by that point.
  }
  \label{fig:ccs_conv}
\end{figure}

\subsection{Non-Cartesian MRI}\label{sec:exp2}
The second experiment is a 2D non-Cartesian variable-density spiral brain
application using data that is publicly available at the code
repository\footnote{\url{https://github.com/sidward/ppcs}}.
Reference data was acquired on a 3T scanner (GE Healthcare, Waukesha, WI) with
IRB approval and informed consent obtained using 32 coils with a variable
density spiral trajectory at 1mm $\times$ 1mm resolution and a field-of-view
of 220mm $\times$ 220mm.
The trajectory was designed with 16-fold in-plane under-sampling at the
center of k-space and a linearly increasing under-sampling rate up to
32-fold at the edge of k-space.
The readout duration was 6.7ms.
This reference data was SVD coil-compressed to 12 coils, and then
retrospectively under-sampled by discarding every other interleave out of 32
interleaves.

The unconstrained reconstruction formulation \cref{eq:ucopt} for this
experiment is as follows:
\begin{equation}\label{eq:exp2}
\begin{aligned}
x^\star &=& \left\{\underset{x}{\text{argmin}} \; \frac{1}{2}
        \lVert \mathcal{F}S x - b \rVert_2^2 + \lambda \lVert Wx\rVert_1\right.
\end{aligned}
\end{equation}
Here, $W$ is the forward Daubechies-4 Wavelet transform, $S$ is the SENSE
model of the parallel-imaging acquisition \cite{pruessmann1999} estimated
using \cite{uecker2014}, and $\mathcal{F}$ is the non-uniform Fourier
transform.
This experiment utilized the Toeplitz structure of $\mathcal{F}^*\mathcal{F}$
for faster evaluation for both \cref{eq:ucopt} and \cref{eq:pucopt}%
\cite{wajer2001, fessler2005, baron2018}.

This application was solved with FISTA, IFISTA, ADMM and FISTA with the
polynomial preconditioner.
The regularization values $(\lambda, \rho)$ tested were the same as in
\Cref{sec:exp1}.
The maximum number of $A^*A$ set to $80$.
For IFISTA, ADMM and the polynomial preconditioned FISTA, the number of
$A^*A$ evaluations were varied as:
\begin{equation}
A^*A \text{ evaluations per iteration} \in \left\{2, 4, 5, 8, 10\right\}
\label{eq:aha2}
\end{equation}
The number of $A^*A$ evaluations per iteration were chosen to cleanly divide
the maximum number of $A^*A$ evaluations for fair comparison between methods.

The same programmatic procedure as in \Cref{sec:param} was used to select
representative parameters for each method with the best FISTA reconstruction
achieving an NRMSE of $\epsilon_f=16.03\%$.
The convergence curves of the respective methods after the programmatically
chosen parameters with respect to both, the number of $A^*A$ evaluations and
observed wall times, are depicted in \Cref{fig:ncr_conv}, and the final
iterates of the corresponding methods are depicted in \Cref{fig:ncr_recons}.
\Cref{fig:ncr_trj} shows the utilized trajectory in \cref{eq:exp2}.

\begin{figure}[ht]
  \centering
  \includegraphics[width=\textwidth]{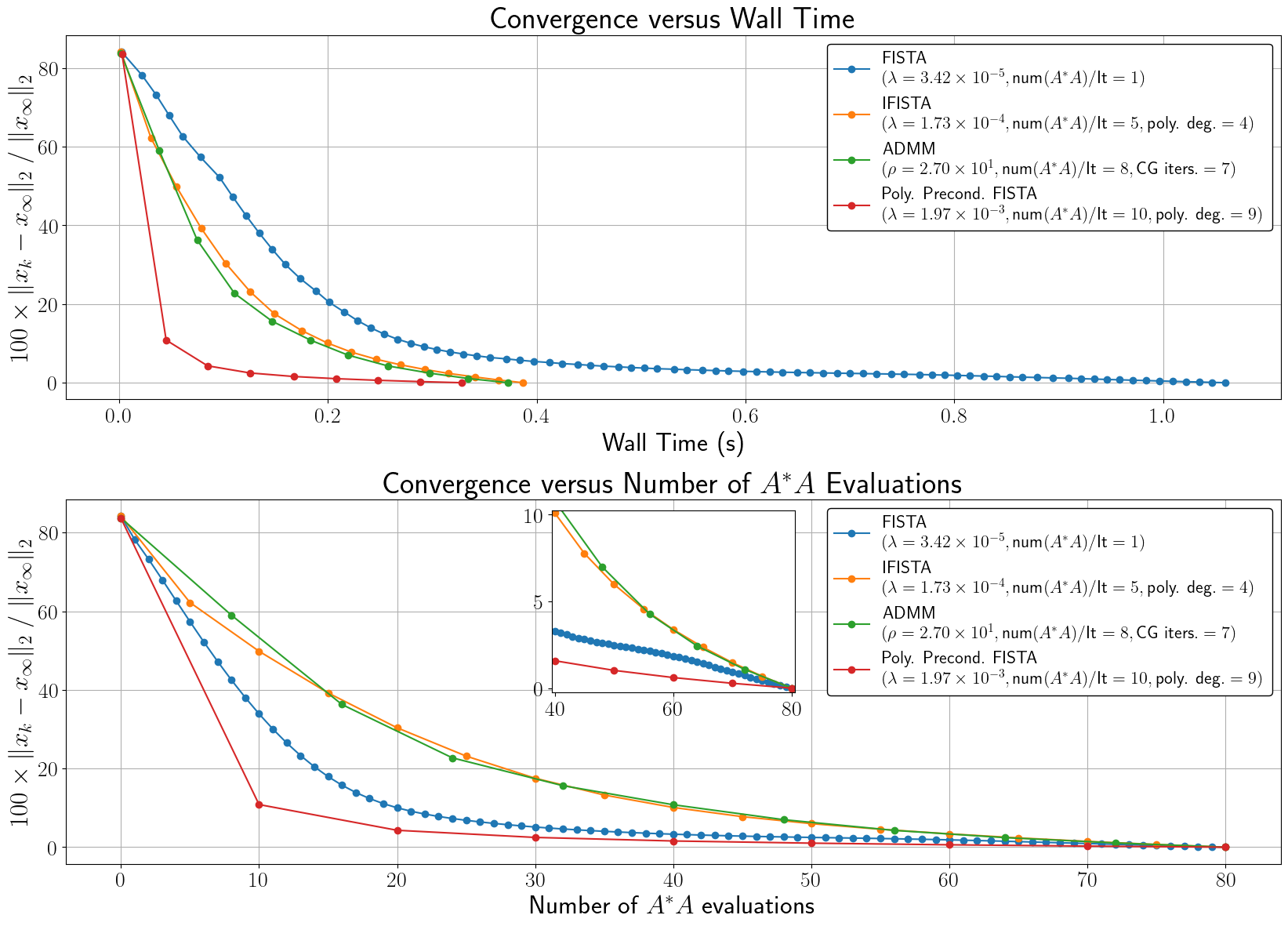}
  \caption{\textbf{Non-Cartesian MRI Convergence Results.}
    This figure depicts the convergence results of the respective methods
    described in \Cref{sec:exp2} given the hyper-parameters selected as in
    \Cref{sec:param} with $\epsilon_f=16.03\%$.
    Given the chosen hyper-parameters, the error over iterations with respect
    to the last iteration of each respective method is plotted.
    The $k^{th}$ iteration and last iteration are labelled as $x_k$ and
    $x_\infty$ respectively.
    The legend on the top right depicts the algorithm and the chosen
    hyper-parameters.
    The x-axis of the top and bottom subplots denotes the total number of
    $A^*A$ evaluations and measured wall-times respectively.
    The circular-markings on each line denote the respective iteration points.
    The number of proximal operators evaluated by a point on the x-axis is
    equal to the number of iterations by that point.
  }
  \label{fig:ncr_conv}
\end{figure}

\subsection{Spatio-Temporal MRI}\label{sec:exp3}
The third experiment is a 3D non-Cartesian spatio-temporal brain application.
The data were obtained on a 3T scanner (GE Healthcare, Waukesha, WI) with IRB
approval and informed consent obtained.
The reference data were acquired with 48 channel coils using a 3D
MRF \cite{ma2013} acquisition with the TGAS-SPI trajectory with a total of 48
groups acquired to achieve adequate 3D k-space encoding at each temporal data
point \cite{cao2022}.

Each acquisition group contains an adiabatic inversion preparation with
TI of 15ms and 500 variable flip-angle acquisitions with TR of 12.5 ms and TE
of 1.75ms, with a 1.2s wait time for signal recovery to improve the
signal-to-noise ratio, resulting in a net acquisition per group of 7.45s and
total acquisition time of approximately 6 minutes.
Additionally, a water-exciting rectangular pulse with duration of
2.38ms was used to depress the fat signal \cite{norbeck2020}.
A variable density spiral trajectory with a 16-fold in-plane
under-sampling rate at the center of k-space and a 32-fold under-sampling
rate at the edge of k-space was used to achieve an encoding at 1-mm
isotropic resolution with a field-of-view of
$220\times220\times220\,\text{mm}^3$ with a readout duration of 6.7 ms.
This reference data was coil-compressed to 10 coils (from 48) using
a combination of \cite{kim2021} and SVD coil compression.
Retrospective under-sampling was performed to simulate a 1 minute
acquisition.
The data used is publicly available at the code
repository\footnote{\url{https://github.com/sidward/ppcs}}.

The unconstrained reconstruction formulation \cref{eq:ucopt} for this
experiment is as follows
\cite{liang2007, petzschner2011, merry2015, zhao2018, tamir2017, cao2022}:
\begin{equation}
x^\star = \left\{\underset{x}{\text{argmin}} \; \frac{1}{2}
      \lVert D^{1/2}\left(\mathcal{F} S \Phi x - b \right)\rVert_2^2 + \lambda
      \text{LLR}(x) \right.
\label{eq:exp3}
\end{equation}
Here, LLR is the locally-low rank constraint \cite{tamir2017} with a block-size
of $8\times8\times8$, $S$ is the SENSE model of the parallel-imaging
acquisition \cite{pruessmann1999} estimated using \cite{ying2007} and
$\mathcal{F}$ is the 3D non-uniform Fourier transform.
Note that in this case, the reconstructed $x$ consists of multiple
``coefficient'' images such that $\Phi x$ recovers the temporal
evolution of the underlying signal.
Please see \cite{liang2007, ma2013, tamir2017, cao2022, wang2019, huang2013,
zhao2015, velikina2015, ben2015, velikina2013} for more information.
\cref{eq:exp3} leverages structural left-preconditioning $D$ for faster
convergence.
The particular method used to derive $D$ was \cite{pipe1999}.
Since the structural left-preconditioner $(D)$ derived from \cite{pipe1999}
differs between the 6-minute acquisition and 1-minute acquisition, the
1-minute FISTA reconstruction of the data is used as a reference for fair
comparison.

The inclusion of $D$ reduces the theoretical efficacy of the polynomial
preconditioner as the resulting normal operator when solving \cref{eq:exp3}
(i.e. $A^*DA$) has a much narrower eigenvalue spectrum compared to $A^*A$.
That being said, utilizing the polynomial preconditioner (which can be directly
applied to \cref{eq:exp3} thanks to its generalizability) can still be
beneficial to reduce real world reconstruction times.
In particular, given the same number of $A^*A$ evaluations, the polynomial
preconditioner achieves qualitatively similar reconstruction to the non
polynomial preconditioned result while utilizing fewer proximal calls.
For applications where evaluating $A^*A$ is much faster than a proximal
call (such as in the sequel), the fewer proximal evaluations yields significant
reductions in processing times.

For this TGAS-SPI-MRF application, \cref{eq:exp3} is too big to solve directly
on even high-end GPUs without utilizing some form of ``batching'' that moves
data between CPU memory and GPU memory when evaluating matrix-vector products
of $A$ and $A^*$.
This increases the per-iteration costs.
Instead, an ADMM formulation is leveraged to split \cref{eq:exp3} into smaller
sub problems, where each sub problem completely fits (that is to say,
there is no need to implement any kind of batching) in approximately 20GB of
GPU memory.
Each sub problem can then be solved in parallel on separate GPU devices with
significantly fewer data transfers between CPU and GPU devices.

The sequel uses 3 GPU devices, but the available code%
\footnote{\url{https://github.com/sidward/ppcs}} generalizes to the number of
available GPUs.
Let $(A_1, b_1)$ denote the forward model and acquired data of the first three
(of ten) receiver channels, $(A_2, b_2)$ denote the next three, and
$(A_3, b_3)$ denote the following three.
The last coil is discarded for now so that each device processes an equal
number of channel data, but future work will be to explore better ways to
integrate all available coils when the number of GPU devices do not cleanly
divide the number of coils.

Sub problems $f_1, f_2$ and $f_3$ are then defined as:
\begin{equation}\begin{array}{c}
f_i=\frac{1}{2}\lVert A_ix - b_i\rVert_2^2 + \frac{\lambda}{3} LLR(x)\\[0.25cm]
\text{for } i \in \{1, 2, 3\}
\end{array}\end{equation}
Note that the sum of the sub problems, i.e. $f_1 + f_2 + f_3$, equals the
objective of \cref{eq:exp3}.
This can now be solved using Global Consensus ADMM \cite{parikh2014}, listed
as \cref{alg:cadmm}.

\begin{algorithm}
\caption{Global Consensus ADMM \cite{parikh2014}}
\label{alg:cadmm}
\begin{algorithmic}
\STATE{\textit{Inputs:}\begin{itemize}
  \item[-] Sub problems $\{f_1, f_2, \dots, f_N\}$
  \item[-] ADMM step-size $\rho$
  \item[-] $\overline{x} = 0$
  \item[-] $x^i = 0$
  \item[-] $u^i = 0$
\end{itemize}}
\STATE{\textit{Step k:} $(k \geq 0)$ Compute
\begin{subequations}
  \begin{equation}\label{eq:cadmm1}
  x_{k+1}^i = \text{\bf prox}_{\rho f_i} \left(\overline{x}_k - u_k^i\right)
  \end{equation}
  \begin{equation}\label{eq:cadmm2}
  \overline{x}_{k + 1} = \frac{1}{N} \sum_i x_{k + 1}^i
  \end{equation}
  \begin{equation}\label{eq:cadmm3}
  u_{k+1}^i = u_k^i + x_{k+1}^i - \overline{x}_{k + 1}^i
  \end{equation}
\end{subequations}}
\end{algorithmic}
\end{algorithm}

With $\{f_i\}$ as the sub problems, \cref{alg:cadmm} can be used to solve
\cref{eq:exp3}.
\cref{eq:cadmm1} and \cref{eq:cadmm3} can be evaluated independently of each
other in parallel on separate devices, with data movement between CPU and GPU
occurring only during \cref{eq:cadmm2}.
\cref{eq:cadmm1} yields a regularized linear inverse problem similar to
\cref{eq:exp3}, and can thus be solved with FISTA and the polynomial
preconditioner.
When solving this sub problem, $A_i$ is much faster to evaluate compared to the
proximal operator of the LLR regularization, making the proximal evaluation
the bottleneck.
Utilizing polynomial preconditioning allows for similar quality reconstruction
given the same number of $A^*A$ evaluations but with fewer proximal calls,
resulting in faster processing times.

Given the size of the problem, the following experiment only considers FISTA
with and without the proposed preconditioner.
The maximum number of normal evaluations was set to 40.
The ADMM step size was set to $1 \times 10^3$.
For both, the standard FISTA and the polynomial preconditioned FISTA
implementations of \cref{eq:cadmm1}, the respective $\lambda$ and
$\lambda_p$ values were qualitatively tuned for the best reconstruction
performance.
The respective convergence curves when solving sub problem associated with
$f_1$ during the first ADMM iteration is depicted in \Cref{fig:mrf_conv}, and
the final reconstruction after 2 ADMM iterations are depicted in
\Cref{fig:mrf_recons}.
The polynomial preconditioning resulted in an approximately $2\times$ faster
reconstruction.

With respect to real word performance, the synchronization step
\cref{eq:cadmm2} took approximately 26 seconds.
Thus the polynomial preconditioned reconstructions depicted in
\Cref{fig:mrf_recons} where achieved in just over 10 minutes.

\begin{figure}[p]
  \centering
  \includegraphics[width=\textwidth]{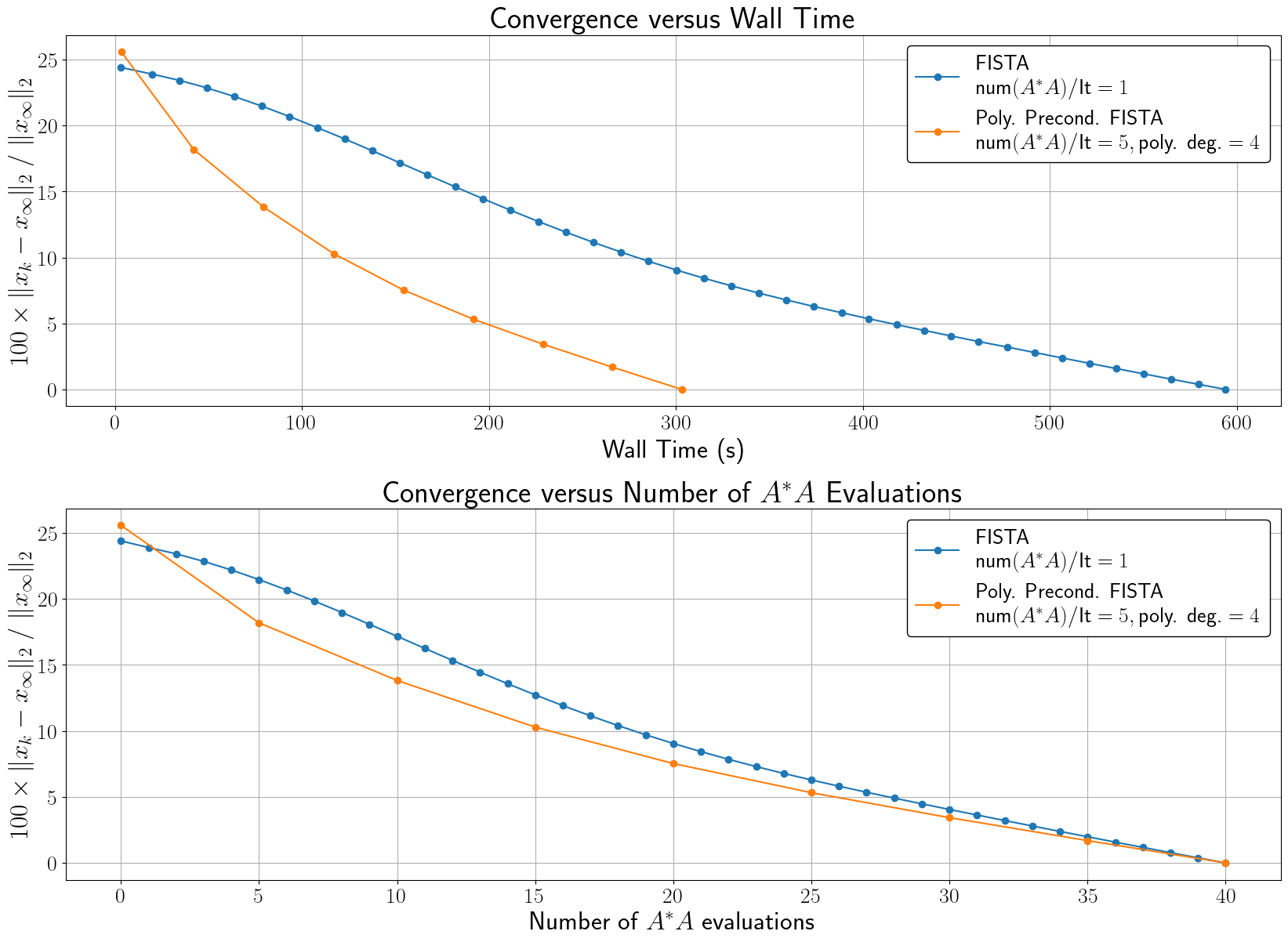}
  \caption{\textbf{Spatio-Temporal MRI Convergence Results.}
    This figure depicts the convergence results of evaluating \cref{eq:cadmm1}
    for the first sub problem $f_1$.
    Given the chosen hyper-parameters, the error over iterations with respect
    to the last iteration of each respective method is plotted.
    The $k^{th}$ iteration and last iteration are labelled as $x_k$ and
    $x_\infty$ respectively.
    The legend on the top right depicts the algorithm and the chosen
    hyper-parameters.
    The x-axis of the top and bottom subplots denotes the total number of
    $A^*A$ evaluations and measured wall-times respectively.
    The circular-markings on each line denote the respective iteration points.
    The number of proximal operators evaluated by a point on the x-axis is
    equal to the number of iterations by that point.
  }
  \label{fig:mrf_conv}
\end{figure}

\section{Discussion}
The polynomial preconditioner is seen to improve the conditioning of the
unconstrained formulation \cref{eq:ucopt}, resulting in faster convergence
compared to standard FISTA, IFISTA and ADMM.
With appropriate tuning of $\lambda_p$ in \cref{eq:pucopt}, the reconstructed
images with the preconditioner are qualitatively similar to the solutions of
\cref{eq:ucopt}.

By utilizing the permutability provided by polynomials and linear operators,
the polynomial preconditioner is applied to the iterates directly instead
of in the range space of $A$, thus significantly reducing the computational
requirements (assuming $n$ is much smaller than $m$, which is often the case
in computational MRI).
Consequently, none of the preconditioned unconstrained reconstructions utilized
any dual variables or an application-specific preconditioner array, resulting
in comparable computational memory requirements to FISTA.
Additionally, it enables the use of the faster $A^*A$ evaluations if
applicable, such as in \Cref{sec:exp2}, which is not possible with
\cite{ong2020}.

While the proposed preconditioner does increase the per-iteration
computational cost compared to FISTA, the proposed method enables faster
computation compared to FISTA even when requiring that both methods to utilize
an equal number of  $A^*A$ evaluations.
This is because, on-top of the theoretically faster convergence offered by
the polynomial preconditioner (as verified by the convergence curve as a
function of the number of $A^*A$ evaluations), utilizing the preconditioner
enables fewer proximal operator evaluations while still achieving comparable
reconstruction.
In particular, for the experiments showcased in this manuscript, the
proximal operator is more computational expensive compared to evaluating
$A^*A$, translating to significantly faster real-world performance as
demonstrated by the convergence curve versus wall-time, while achieving
similar reconstruction performance.
The reduced number of proximal evaluations of the proposed method will be of
even more benefit for highly intensive proximal proximal operators, such as the
structured matrix completion approaches akin to the soft-thresholding versions
of \cite{haldar2013, shin2014}.

The generalizability of the polynomial preconditioner allows it to be used
directly in a subspace-reconstruction \cite{liang2007, petzschner2011,
merry2015, zhao2018, tamir2017, cao2022} without needing to explicitly account
for $\Phi$ when designing the preconditioner while also easily integrating the
structure-specific left-preconditioning ($D$ in \cref{eq:exp3}).
By utilizing $A^*A$ to construct the preconditioner, the method inherently
takes into account information from $\Phi, S, \mathcal{F}$ and $\mathcal{D}$
without user modification.
In particular, as depicted in \Cref{fig:mrf_recons}, the combination of
$\mathcal{D}$, the ADMM parallel processing and the polynomial preconditioner
enabled a high quality reconstruction in just over 10 minutes, which is an
approximately $12\times$ improvement over the computation times reported in
\cite{cao2022}.

Looking at the convergence of iterates as a function of the number of $A^*A$
evaluations, IFISTA performs slightly worse than standard FISTA in
\Cref{sec:exp1} and \Cref{sec:exp2}.
This is arguably because, as discussed in \Cref{sec:complex}, the IFISTA
preconditioner cannot be pre-calculated and is instead evaluated in a
matrix-free manner.
In contrast, by explicitly optimizing for the faster convergence of iterates
for a given degree via \cref{eq:l2cost}, the proposed polynomial preconditioner
enables faster convergence.

This work does not explore incorporating a weighting-prior into the spectral
cost ($w$ in \cref{eq:wl2cost}).
It is expected that a reasonable prior estimate of the spectrum of an
operator $A$ derived for a specific application will significantly
improve the rate of convergence for that application.
However, given the large dimensionality of $A$ and the consequent difficulty in
approximating the spectrum of $A^*A$, estimating a reasonable prior $w$ will
involve trial-and-error, and is thus left to future work.

A limitation of the polynomial preconditioner is with respect to numerical
stability.
In principle, it is possible to utilize a polynomial $p$ of a high degree $d$.
However, in practice, evaluating powers of $A^*A$ can accumulate numerical
errors.
Therefore, tuning the degree $d$ for the application of interest is required.

The error bound presented in \cref{thm:err} is applicable to any linear
operator $P$ such that $P$ is injective on $\text{null}(A)^\perp$.
Thus, deriving a $P$ to minimize the error bound in \cref{eq:errbnd} is a
promising avenue for application-specific preconditioner design.
For example, \cref{eq:errbnd} can be used to upper-bound the error of the
circulant preconditioner in \cite{muckley2016}.


\section*{Acknowledgments}
The authors would like to thank Dr.\ Sophie Schauman for testing the
reproducibility code, Dr.\ Martin Uecker for providing additional
references and Dr.\ Tao Hong for the IFISTA reference.

\begin{figure}[ht]\centering
  \includegraphics[width=0.8\textwidth]{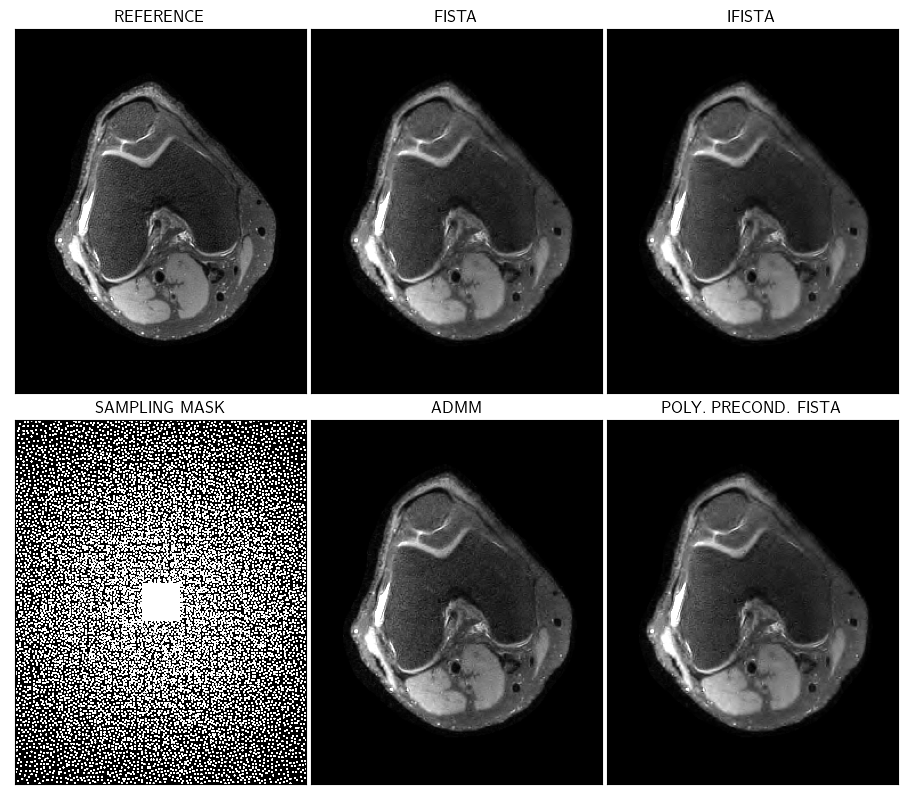}
  \caption{\label{fig:ccs_recons}
    \textbf{Cartesian MRI Reconstruction Results.} This figure depicts the
    final iterations of the respective methods in \cref{fig:ccs_conv}
    The bottom left figure, labelled ``SAMPLING MASK'', denotes the
    under-sampling mask used in \cref{eq:exp1}.
    The hyper-parameters for these results are depicted in
    \cref{fig:ccs_conv}.
  }
  \includegraphics[width=0.8\textwidth]{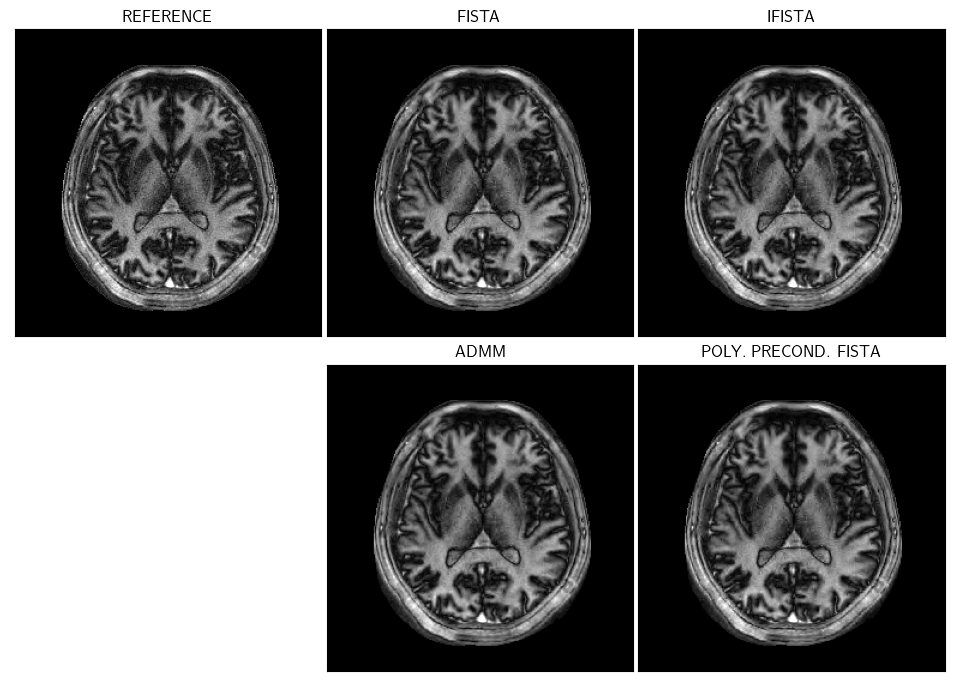}
  \caption{\label{fig:ncr_recons}
    \textbf{Non-Cartesian MRI Reconstruction Results.} This figure depicts the
    final iterations of the respective methods in \cref{fig:ncr_conv}.
    The hyper-parameters for these results are depicted in
    \cref{fig:ncr_conv}.
  }
\end{figure}
\clearpage
\begin{figure}[ht]\centering
  \includegraphics[width=0.8\textwidth]{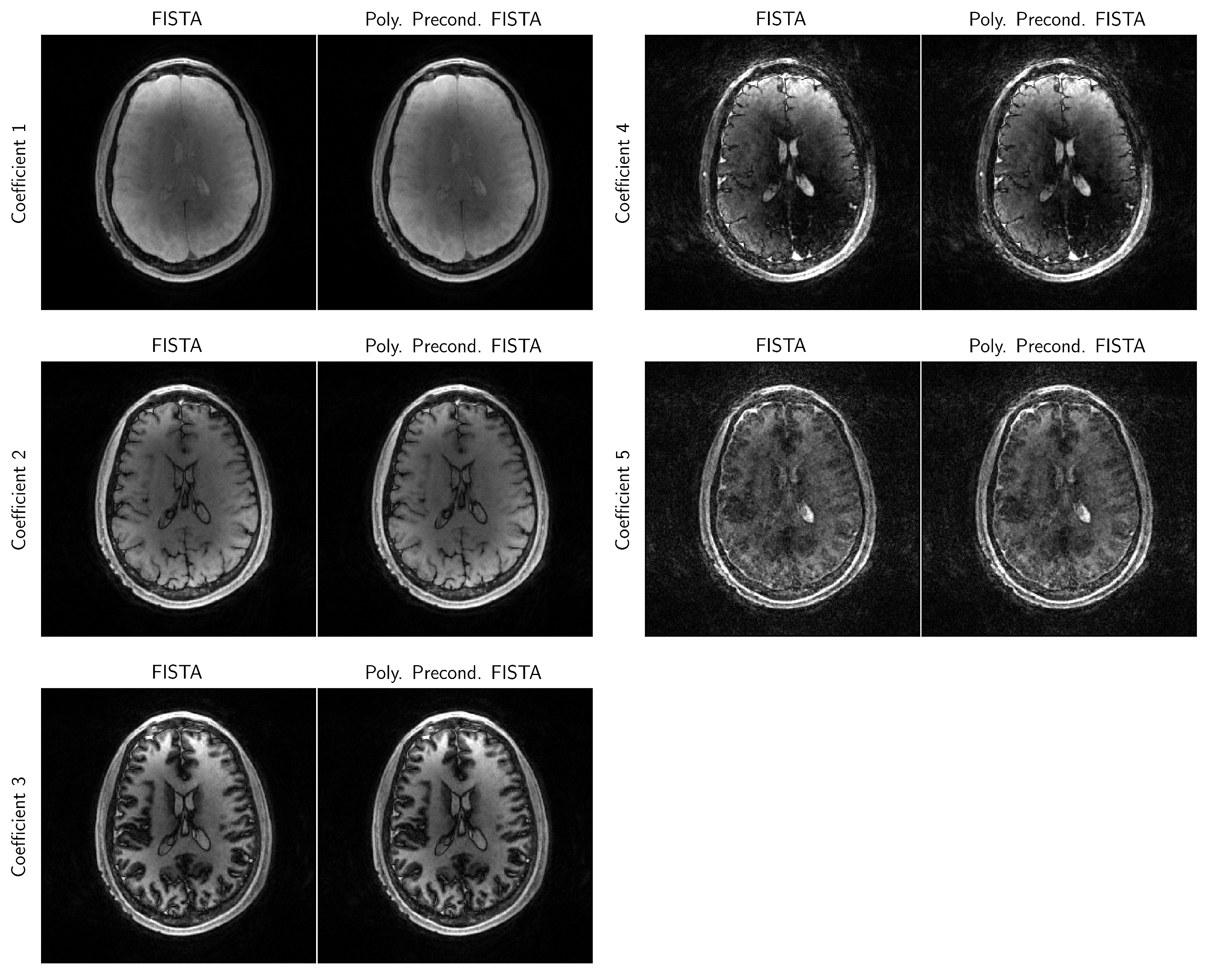}
  \caption{\label{fig:mrf_recons}
    \textbf{Spatio-Temporal MRI Reconstruction Results.} This figure depicts
    the reconstructions after two ADMM iterations, where \cref{eq:cadmm1} is
    evaluated using (left) FISTA and (right) Polynomial Preconditioned FISTA.
    The hyper-parameters for these results are depicted in \cref{fig:mrf_conv}.
  }
\end{figure}
\clearpage
\begin{figure}[ht]\centering
  \includegraphics[width=0.5\textwidth]{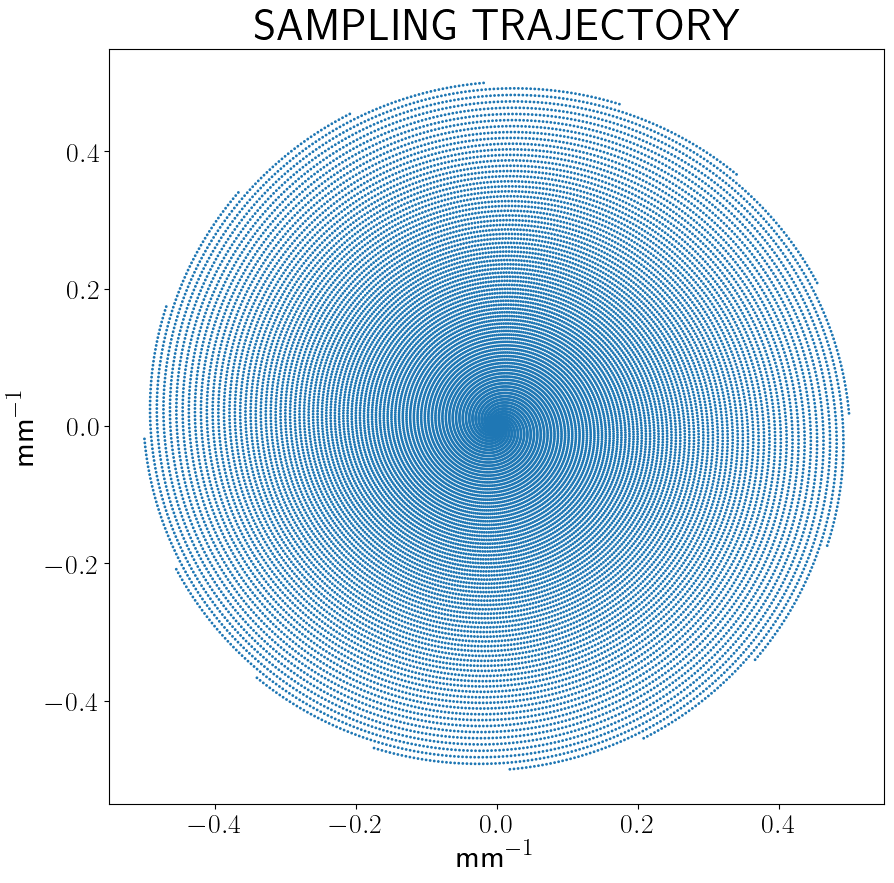}
  \caption{\label{fig:ncr_trj}
    \textbf{Non-Cartesian MRI Trajectory.} This figure depicts the
    under-sampled trajectory used in \cref{eq:exp2}.
  }
\end{figure}

\end{document}